\documentclass[singlecolumn,         
               showpacs,            
               preprintnumbers,     
               aps,                 
               prd,          	    
               a4paper,             
               superscriptaddress,      
               nofootinbib,         
               tightenlines,        
               floats,floatfix,11pt
               ]{revtex4-2}              
\usepackage{graphicx}  
\usepackage[margin=0.8in]{geometry}
\usepackage{dcolumn}   
\usepackage{bm}  
\usepackage[normalem]{ulem}
\usepackage[utf8]{inputenc}
\usepackage{amsmath,amssymb}
\usepackage{amsthm}
\usepackage{soul}
\usepackage{float}
\usepackage[thinlines]{easytable}
\usepackage{array,booktabs}
\usepackage{lipsum} 
\usepackage{multirow}
\usepackage{placeins} 
\usepackage[colorlinks=true,linkcolor=blue,citecolor=blue]{hyperref}
\usepackage{subcaption}
\newtheorem{theorem}{Theorem}

\usepackage{array,makecell}
\usepackage{cleveref} 
\begin{document} 

\title{Hypersurface Anchored Variational Principle for General Relativity}
\author{Prasanta Sahoo }
\email{prasantmath123@yahoo.com} 
\affiliation{Midnapore College (Autonomous), Midnapore, West Bengal, India, 721101}


\begin{abstract}

A hypersurface anchored variational extension of General Relativity is formulated in which the Einstein-Hilbert action is supplemented by a diffeomorphism invariant functional supported on an embedded spacelike hypersurface whose embedding is varied independently of the spacetime metric. The resulting Euler-Lagrange system consists of the Einstein equations with a localized distributional contribution together with an anchoring equation determining admissible embeddings. For the admissible class of hypersurface functionals considered here, the bulk field equations retain the standard second order principal structure away from the hypersurface and no additional propagating bulk gravitational degrees of freedom are introduced. Under ellipticity and invertibility assumptions, local persistence and linear stability of anchoring hypersurfaces follow from standard implicit function and elliptic estimates. The anchoring condition is generically inequivalent to local slicing gauge conditions and therefore defines a genuine variational restriction rather than a coordinate choice. In the canonical formulation, the momentum constraints retain their standard form, whereas the Hamiltonian constraint acquires a hypersurface supported term. The corresponding smeared generators close in the weak sense: the Dirac algebra is recovered in the bulk, with deviations confined to localized distributional surface contributions. The construction therefore defines a constrained sector of the classical solution space of General Relativity consisting of spacetimes that admit at least one embedded hypersurface satisfying the anchoring equation. In homogeneous cosmology, the localized term induces matching conditions across the anchoring surface, allowing finite transitions in the expansion rate while preserving standard evolution away from the transition hypersurface.

\end{abstract}

\maketitle

\let\oldthefootnote\thefootnote
\renewcommand{\thefootnote}{\textit{}}
\footnotetext{Corresponding author: \href{mailto:prasantmath123@yahoo.com}{prasantmath123@yahoo.com}}
\let\thefootnote\oldthefootnote

\section{Introduction}

General Relativity describes gravitation as the dynamics of spacetime geometry governed by Einstein's field equations \cite{Einstein1915,Wald1984,Carroll2004,PoissonWill2014}. In its standard formulation, the theory admits a well posed Cauchy initial value problem: initial data specified on a spacelike hypersurface and satisfying the Hamiltonian and momentum constraints determine a unique maximal spacetime development up to diffeomorphism \cite{ChoquetBruhatGeroch1969,Ringstrom2009,Friedrich1985}. This formulation underlies most analytical and numerical investigations of relativistic systems and provides a complete local dynamical description. While the Cauchy formulation fully characterizes local evolution, it does not impose restrictions on the global class of admissible solutions beyond consistency with the field equations and constraints. In applications such as cosmology, physically relevant solutions correspond to a restricted subset of the space of admissible initial data, typically selected by additional assumptions on symmetry, regularity, or phenomenological input \cite{Penrose1979,Weinberg1989,Carroll2001}. More generally, several properties of gravitational spacetimes, including the presence of horizons, causal boundaries, and singularity structure, depend on global geometric features that are not determined solely by local initial data \cite{Penrose1965}. These observations motivate the consideration of frameworks in which global geometric conditions are incorporated directly at the level of the variational principle.

At the level of the action, General Relativity already exhibits sensitivity to global geometric structure. The Einstein-Hilbert action requires boundary contributions, most notably the Gibbons-Hawking-York term, to define a well posed variational principle \cite{York1972}. Hypersurfaces also play a role in several related constructions, including thin-shell and junction formalisms \cite{BarrabesIsrael1991}, geometric treatments of embedded hypersurfaces \cite{MarsSenovilla1993}, higher curvature extensions \cite{Myers1987,GravanisWillison2003}, and brane-world scenarios \cite{RandallSundrum1999a,RandallSundrum1999b,Dvali2000,Charmousis2002}. In these settings, hypersurfaces are introduced either as boundaries, as loci of distributional stress-energy, or as externally specified geometric structures. The canonical (ADM) formulation further emphasizes the role of hypersurfaces in the dynamical description. The Hamiltonian structure encodes the dynamics in terms of constraint equations generating diffeomorphism invariance \cite{Dirac1964,ReggeTeitelboim1974,Teitelboim1973,Kuchar1976}. These constraints govern the evolution of initial data but do not restrict the space of solutions beyond consistency with the field equations.

In the present work, a variational framework is developed in which an embedded spacelike hypersurface enters as an independent geometric variable of the action. The Einstein-Hilbert functional is supplemented by a diffeomorphism invariant hypersurface term whose embedding variables are varied independently of the spacetime metric. The resulting Euler-Lagrange system determines both the metric and admissible embeddings, so that the hypersurface is not prescribed a priori and is not identified with a foliation choice.

The hypersurface contribution induces a localized distributional term supported on the dynamically selected hypersurface. Away from this support, the equations reduce exactly to the Einstein equations, preserving the principal part and second order character of the bulk dynamics. Under suitable regularity assumptions, the coupled system is well defined in the weak sense and introduces no additional propagating bulk gravitational degrees of freedom. Instead, admissible solutions are restricted to spacetimes admitting at least one hypersurface satisfying the associated embedding equation.

For the admissible class of functionals considered here, the embedding equation is elliptic. Under standard invertibility assumptions on the linearized operator, local persistence and linear stability of anchoring hypersurfaces follow from standard implicit function and elliptic estimates.

The construction is distinct from boundary counterterms, thin-shell junction conditions, and gauge slicings such as constant mean curvature foliations \cite{Isenberg1978}. The anchoring condition arises from an independent variational equation and therefore defines a genuine constrained sector of General Relativity rather than a coordinate prescription.

At the canonical level, the momentum constraints retain their standard form, while the Hamiltonian constraint acquires a hypersurface supported contribution. The smeared generators close in the weak sense: the Dirac algebra is recovered in the bulk, with corrections confined to localized distributional surface terms. The embedding variable therefore acts as a quasi-local constraint variable rather than as an additional propagating mode.

As an illustration, homogeneous and isotropic cosmology is considered. The hypersurface term induces matching conditions across the anchoring surface, allowing finite discontinuities in the expansion rate while preserving standard evolution away from the transition hypersurface.

The paper is organized as follows. Section~\ref{sec:Structural_Limitations} discusses limitations of purely initial-value formulations. Section~\ref{sec:variational} develops the variational principle and field equations. Section~\ref{sec:Mathematical_Structure_and_Consistency} establishes consistency and existence results. Section~\ref{sec:adm} presents the Hamiltonian formulation and weak constraint algebra. Section~\ref{sec:Relation_to_Existing_Geometric_Constructions} analyzes the relation to existing constructions and the restricted solution sector. Section~\ref{sec:Perturbation} studies linear perturbations. Section~\ref{sec:Quantum_and_Thermodynamic_Implications} gives semiclassical remarks. Section~\ref{sec:conclusion} concludes.

\section{Structural Limitations of Initial Value Formulations}
\label{sec:Structural_Limitations}

General Relativity admits a well defined formulation as a Cauchy initial value problem through the Arnowitt-Deser-Misner (ADM) decomposition \cite{York1979,FischerMarsden1979}. In this formulation, initial data specified on a spacelike hypersurface $\Sigma_0$, subject to the Hamiltonian and momentum constraints, determine a unique maximal globally hyperbolic spacetime development up to diffeomorphism \cite{BartnikIsenberg2004}. This framework provides the basis for both analytical and numerical studies of relativistic spacetimes.

While the initial value formulation provides a complete local dynamical description, it does not impose restrictions on the global class of admissible solutions beyond consistency with the field equations and constraints. In cosmological applications, observationally relevant spacetimes correspond to a restricted subset of admissible initial data, typically selected by additional assumptions such as symmetry, regularity, or phenomenological input \cite{Penrose1989,CarrollChen2004}. This issue is closely related to broader problems in cosmology, including the cosmological constant problem and the origin of dark energy \cite{Weinberg1989,Joyce2015,Copeland2006}.

From a statistical perspective, attempts to characterize typical solutions using measures on phase space encounter well known ambiguities. Analyses based on Liouville type measures indicate that subsets of solutions exhibiting specific large scale properties, such as approximate homogeneity and isotropy, occupy a restricted region of the admissible configuration space \cite{RemmenCarroll2013}. More generally, defining physically meaningful probability measures on the space of cosmological solutions remains a nontrivial problem \cite{Freivogel2011,Guth2007,Ellis2012}. These considerations suggest that the field equations alone do not provide a mechanism for selecting physically distinguished solutions within the full solution space.

Several fundamental features of gravitational physics further illustrate the role of global structure. Phenomena such as event horizons, causal boundaries, and spacetime singularities are defined in terms of extended regions and global causal relations \cite{Geroch1970,Senovilla1998}. Although compatible with the Cauchy formulation, their characterization cannot be reduced to local data specified on a single hypersurface. This highlights the distinction between local dynamical evolution and global geometric properties of spacetime.

The thermodynamic arrow of time provides an additional example often discussed in this context. While the underlying dynamical laws are approximately time reversal invariant, observed cosmological evolution exhibits a pronounced temporal asymmetry. Within the standard initial value framework, this asymmetry is typically associated with special initial conditions imposed on a spacelike hypersurface \cite{Lebowitz1993,Callender2010}. From a statistical standpoint, low entropy configurations correspond to a restricted region of phase space with respect to Liouville type measures. However, a complete derivation of macroscopic irreversibility in gravitational systems remains an open problem, and no definitive selection principle is provided by the classical field equations alone.

These observations motivate the exploration of frameworks in which additional global conditions are incorporated at the level of the action, rather than imposed externally as initial or boundary conditions. The objective is not to replace the Cauchy formulation, but to supplement it with a mechanism that restricts the admissible class of solutions in a diffeomorphism invariant manner. In the following section, a hypersurface anchored variational principle is introduced, in which a dynamically determined embedded hypersurface arises from the variational problem and imposes a global constraint on admissible spacetime geometries.

\section{Hypersurface Anchored Variational Principle and Field Equations}
\label{sec:variational}

Motivated by the considerations of Sec.~\ref{sec:Structural_Limitations}, a variational framework is formulated in which an embedded spacelike hypersurface enters as an independent geometric variable. The configuration space consists of pairs $(g_{\mu\nu},X^\mu)$, where $g_{\mu\nu}$ is the spacetime metric and $X^\mu$ specifies the embedding of the hypersurface \cite{Gourgoulhon20072}.

Let $\mathcal M$ be a four dimensional Lorentzian manifold with metric $g_{\mu\nu}$ of regularity $C^2$ on $\mathcal M\setminus\Sigma_\star$, where $\Sigma_\star\subset\mathcal M$ is a smooth codimension one embedded spacelike hypersurface. The embedding is given by $X^\mu:\Sigma\rightarrow \mathcal M,~~ y^a\mapsto X^\mu(y^a)$, where $y^a$ $(a=1,2,3)$ are intrinsic coordinates on an abstract three manifold $\Sigma$. The induced metric and extrinsic curvature are
\begin{equation}
h_{ab}=g_{\mu\nu}\,\partial_aX^\mu\partial_bX^\nu,
\qquad
K_{ab}=h_a{}^\mu h_b{}^\nu \nabla_\mu u_\nu,
\label{eq:induceddata}
\end{equation}
where $u^\mu$ is the future directed unit normal satisfying $u_\mu u^\mu=-1$, and $h^\mu{}_\nu=\delta^\mu{}_\nu+u^\mu u_\nu$ is the tangential projector \cite{Poisson2004}. The action is taken to be
\begin{align}
S[g,X] & = \frac{1}{16\pi G}\int_{\mathcal{M}} d^{4}x\,\sqrt{-g} R +\lambda \mathcal{I}[\Sigma_{\star},g,X]  +S_{\rm m}[g,\Psi],
\label{eq:totalaction}
\end{align}
where $\lambda$ is a constant parameter and $\mathcal I$ is a hypersurface supported diffeomorphism invariant functional.

\subsection{Admissible hypersurface functionals}

Throughout the present work, the rigorous admissible class is restricted to functionals of the form
\begin{equation}
\mathcal{I}[\Sigma_{\star},g,X] = \int_{\Sigma_{\star}} d^{3}y \sqrt{h} 
F\left(h_{ab},K_{ab},\mathcal{R}[h]\right),
\label{eq:Iadmissible}
\end{equation}
where $\mathcal R[h]$ is the intrinsic Ricci scalar of $(\Sigma_\star,h_{ab})$, and $F$ is a smooth local scalar depending algebraically on its arguments. Hypersurface actions built from intrinsic and extrinsic geometry occur widely in gravitational boundary value problems and higher curvature theories. The following assumptions are imposed: $(i)$ $F$ contains no normal derivatives of $K_{ab}$; $(ii)$ $F$ contains no pullback of four dimensional curvature tensors such as $R[g]$; $(iii)$ under metric variation at fixed embedding, all derivatives of $\delta g_{\mu\nu}$ are tangential and removable by intrinsic integration by parts on $\Sigma_\star$; and $(iv)$ under normal embedding variation, the principal part of the Euler-Lagrange operator is elliptic on $\Sigma_\star$.

These conditions guarantee that the hypersurface term does not alter the bulk principal symbol of the Einstein equations. Extensions involving pullbacks of four dimensional curvature invariants require additional boundary analysis and are not considered in the rigorous formulation developed here. A representative admissible choice is
\begin{equation}
F=K+\alpha \mathcal{R}[h]+\beta K_{ab}K^{ab},
\label{eq:representativeF}
\end{equation}
with constants $(\alpha,\beta)$.

\subsection{Distributional representation}

The hypersurface functional may be written equivalently as
\begin{equation}
\mathcal{I}= \int_{\mathcal{M}} d^{4}x\,\sqrt{-g}\,\delta(\Sigma_{\star})\,\mathcal{L}_{\Sigma},
\label{eq:distributionalI}
\end{equation}
where $\mathcal{L}_{\Sigma}=\sqrt{h} F$ and $\delta(\Sigma_{\star})$ is the Dirac distribution defined by
\begin{equation}
\int_{\mathcal{M}} d^{4}x\,\sqrt{ -g} \delta(\Sigma_{\star}) \varphi
=
\int_{\Sigma_{\star}} d^{3}y \sqrt{h} \varphi(X(y))
\label{eq:deltaSigma}
\end{equation}
for every smooth compactly supported test function $\varphi$. Distributional hypersurface sources of codimension one are standard in gravitational junction theory \cite{Israel1966}. All expressions are interpreted in the weak sense. Since $\delta(\Sigma_\star)$ multiplies only quantities regular on $\Sigma_\star$, no products of singular distributions occur.

\subsection{Independent variations}

The variational principle is defined by $\delta S=0$ under independent variations of $g_{\mu\nu}$ and $X^\mu$. Embedding variations are decomposed as $\delta X^\mu=\Phi u^\mu+\xi^a\partial_aX^\mu$. Tangential deformations $\xi^a\partial_aX^\mu$ correspond to reparametrizations of $\Sigma_\star$ and therefore do not contribute to the Euler-Lagrange equation. The physical embedding variation is the normal deformation $\delta X^\mu=\Phi u^\mu$ \cite{MarsSenovilla1993}.

\subsection{Metric variation}

At fixed embedding,
\begin{equation}
\delta_{g} S = \frac{1}{16\pi G}\int_{\mathcal M} d^{4}x \sqrt{-g} G_{\mu\nu}\,\delta g^{\mu\nu}
+\lambda \delta_g\mathcal I
+\delta_g S_{\rm m}.
\label{eq:metricvariation1}
\end{equation}

For every admissible functional \eqref{eq:Iadmissible}, intrinsic integrations by parts on $\Sigma_\star$ yield
\begin{align}
\delta_g\mathcal I &=
\frac12\int_{\Sigma_\star} d^3y \sqrt h \tau^{ab} \delta h_{ab} = \int_{\mathcal M} d^4x \sqrt{-g} \delta(\Sigma_\star) \tau_{\mu\nu} \delta g^{\mu\nu},
\label{eq:metricvariationI}
\end{align}
where $\tau_{\mu\nu}=h_\mu{}^a h_\nu{}^b \tau_{ab}$ is tangential to $\Sigma_\star$. No second normal derivatives of $\delta g_{\mu\nu}$ arise. Hence the field equations are
\begin{equation}
G_{\mu\nu} +\lambda \delta(\Sigma_\star)\tau_{\mu\nu} = 8\pi G T_{\mu\nu},
\label{eq:modifiedeinstein}
\end{equation}
in the sense of distributions.

Away from $\Sigma_\star$, the localized term vanishes and the equations reduce exactly to Einstein's equations. Therefore the bulk principal part remains that of General Relativity. The induced hypersurface stress tensor is
\begin{equation}
T^{(\Sigma_\star)}_{\mu\nu} = -\frac{2}{\sqrt{-g}} \frac{\delta\mathcal I}{\delta g^{\mu\nu}},
\label{eq:surfstress}
\end{equation}
with functional derivative understood weakly. For the representative model
\begin{equation}
\mathcal I=\int_{\Sigma_\star} d^3y \sqrt h K,
\label{eq:IKmodel}
\end{equation}
one recovers the standard Brown-York type surface variation \cite{BrownYork1993}.

\subsection{Embedding variation}

Variation with respect to the embedding gives
\begin{equation}
\delta_X\mathcal I = \int_{\Sigma_\star} d^3y \sqrt h E[g,X] \Phi,
\label{eq:embeddingvariation}
\end{equation}
and therefore the anchoring equation
\begin{equation}
E[g_{\mu\nu},X^\mu]=0.
\label{eq:anchoreq}
\end{equation}

This equation is intrinsic to $\Sigma_\star$ and determines admissible embeddings. It is not satisfied by generic hypersurfaces. For functionals depending nontrivially on $K_{ab}$, the normal variation contains the standard term
\begin{equation}
\delta K = -\nabla^a\nabla_a\Phi - \left(K_{ab}K^{ab}+R_{\mu\nu}u^\mu u^\nu\right)\Phi,
\label{eq:deltaK}
\end{equation}
so that the principal part of $E=0$ is governed by the Laplace-Beltrami operator on $\Sigma_\star$ \cite{Andersson2008}. Under the admissibility assumptions, the embedding equation is therefore elliptic.

\subsection{Resulting system}

The coupled Euler-Lagrange system consists of the Einstein equations with a hypersurface supported distributional contribution \eqref{eq:modifiedeinstein}; and the anchoring equation \eqref{eq:anchoreq} for the embedding.
Accordingly, the bulk dynamics remain hyperbolic away from $\Sigma_\star$, whereas the embedding sector is elliptic on $\Sigma_\star$. The hypersurface acts as a quasi local geometric constraint on admissible solutions and introduces no additional propagating bulk degrees of freedom.

\section{Mathematical Structure and Consistency}
\label{sec:Mathematical_Structure_and_Consistency}

The hypersurface anchored construction must satisfy the standard requirements of a constrained variational theory: a well defined Euler-Lagrange system, compatibility with diffeomorphism invariance, preservation of the local propagating content of General Relativity in the bulk, and mathematical consistency of the hypersurface supported source term \cite{HenneauxTeitelboim1992}. These properties are formulated under the admissible class introduced in Sec.~\ref{sec:variational}.

Let $(M,g_{\mu\nu})$ be a four dimensional Lorentzian manifold and let $\Sigma_\star\subset M$ be a smooth embedded spacelike hypersurface specified by an embedding $X^\mu:\Sigma\rightarrow M$, where $\Sigma$ is a smooth three manifold. The configuration space is
\begin{equation}
\mathcal C= \left\{ \left.
\begin{array}{c}
(g_{\mu\nu},X^\mu)\\[1.2ex]
\end{array}
\right|\;
\begin{array}{l}
g_{\mu\nu}\in C^0(M)\cap C^2(M\setminus\Sigma_\star),\\[0.4ex]
X^\mu\in C^2(\Sigma,M)
\end{array}
\right\}.
\label{eq:configspace}
\end{equation}

The metric is assumed continuous across $\Sigma_\star$, with first derivatives piecewise locally bounded on each side. Intrinsic geometric quantities induced on $\Sigma_\star$, including $h_{ab}$ and $K_{ab}$, are therefore finite. The action is given by Eq.~(\ref{eq:totalaction}). The admissible class is chosen so that the hypersurface contribution acts as a localized geometric constraint without modifying the bulk principal dynamics. The defining properties are:

\begin{enumerate}
\item[(i)] Away from $\Sigma_\star$, the field equations reduce exactly to Einstein's equations.

\item[(ii)] The embedding equation is elliptic on $\Sigma_\star$ and introduces no independent hyperbolic sector.

\item[(iii)] Compactly supported variations $(\delta g_{\mu\nu},\delta X^\mu)$ produce a consistent Euler-Lagrange system.

\item[(iv)] The localized term $\delta(\Sigma_\star)\tau_{\mu\nu}$ is meaningful in $\mathcal D'(M)$, and no products of singular distributions occur \cite{GerochTraschen1987}.
\end{enumerate}

A representative admissible model is given by Eq.~(\ref{eq:representativeF}). The hypersurface contribution has the form
\begin{equation}
T^{(\Sigma_\star)}_{\mu\nu}
=
\delta(\Sigma_\star)\tau_{\mu\nu},
\label{eq:surfdiststress}
\end{equation}
where $\tau_{\mu\nu}$ is tangential to $\Sigma_\star$ and regular on the hypersurface.

Since the Einstein equations are nonlinear, distributional sources require controlled regularity. In the present setting, only codimension one Dirac terms multiplying bounded tensor coefficients arise. No quadratic products such as $\delta(\Sigma_\star)^2$ occur. Accordingly, the weak formulation is sufficient. Stronger nonlinear frameworks, such as Colombeau generalized geometry, may also be employed but are not required here \cite{SteinbauerVickers2006}. This construction differs from thin-shell geometries with jump discontinuities in extrinsic curvature, since the localized term originates from the variational functional and the metric itself is taken continuous across $\Sigma_\star$. Under the above assumptions, the stationarity condition $\delta S=0$ yields
\begin{equation}
G_{\mu\nu} +\lambda\,T^{(\Sigma_\star)}_{\mu\nu}
= 8\pi G\,T_{\mu\nu},
\qquad
E[g_{\mu\nu},X^\mu]=0,
\label{eq:weakELsystem}
\end{equation}
in the sense of distributions. The first equation governs the spacetime metric, while the second determines admissible embeddings. Because $T^{(\Sigma_\star)}_{\mu\nu}$ is supported only on $\Sigma_\star$, the principal symbol of the bulk Einstein operator is unchanged in every open set disjoint from $\Sigma_\star$.

The action is invariant under simultaneous diffeomorphisms of $g_{\mu\nu}$ and $X^\mu$. Therefore the Euler-Lagrange system satisfies the weak Noether identity
\begin{equation}
\nabla^\mu \left(G_{\mu\nu} +\lambda T^{(\Sigma_\star)}_{\mu\nu} -8\pi G T_{\mu\nu} \right) =0
\quad \text{in }\mathcal D'(M).
\label{eq:weaknoether}
\end{equation}

Terms involving derivatives of $\delta(\Sigma_\star)$ cancel against the embedding equation $E=0$. Hence no additional independent conservation law is required beyond the coupled Euler-Lagrange system. A natural question is whether anchored hypersurfaces persist under sufficiently small perturbations of a background metric.

\begin{theorem}[Local persistence by the implicit function theorem]
\label{thm:IFT}
Let $(M,g^{(0)}_{\mu\nu})$ be a smooth spacetime admitting a smooth embedded spacelike hypersurface with embedding $X_0^\mu$ satisfying $E[g^{(0)},X_0]=0$. Assume:

\begin{enumerate}
\item[(i)] $E[g,X]$ defines a $C^1$ map 
\[
E:\mathcal U\subset\mathcal G\times\mathcal X\rightarrow\mathcal Y,
\]
where $\mathcal G$ is a neighborhood of $g^{(0)}$ in $C^{k,\alpha}$, $\mathcal X$ is a neighborhood of $X_0$ in $C^{k+2,\alpha}$, and $\mathcal Y=C^{k,\alpha}(\Sigma)$;

\item[(ii)] the linearization with respect to the embedding variable,
\[
D_XE[g^{(0)},X_0]:
C^{k+2,\alpha}(\Sigma)\rightarrow C^{k,\alpha}(\Sigma),
\]
is an elliptic isomorphism.
\end{enumerate}

Then there exists a neighborhood $\mathcal V\subset\mathcal G$ of $g^{(0)}$ and a unique $C^1$ map $g\mapsto X(g)$ such that $E[g,X(g)]=0$ for all $g\in\mathcal V$. The solution is unique modulo tangential reparametrizations.
\end{theorem}

\begin{proof}
Since $E$ is continuously Fr\'echet differentiable and $D_XE[g^{(0)},X_0]$ is an isomorphism, the Banach space implicit function theorem applies \cite{Zeidler1986}. This yields neighborhoods $\mathcal V\subset\mathcal G$ and $\mathcal W\subset\mathcal X$, together with a unique $C^1$ map $\Phi:\mathcal V\to\mathcal W$ satisfying $E[g,\Phi(g)]=0$. Setting $X(g)=\Phi(g)$ proves the result. Tangential reparametrizations act trivially on the geometric hypersurface.
\end{proof}

\begin{theorem}[Distributional consistency]
\label{thm:distributionalconsistency}
Let $(M,g_{\mu\nu})$ satisfy:

\begin{enumerate}
\item[(i)] $g_{\mu\nu}\in C^0(M)\cap C^2(M\setminus\Sigma_\star)$;

\item[(ii)] first derivatives of $g_{\mu\nu}$ are piecewise locally bounded;

\item[(iii)] the induced metric on $\Sigma_\star$ is nondegenerate;

\item[(iv)] $\tau_{\mu\nu}\in L^\infty(\Sigma_\star)$.
\end{enumerate}

Then:

\begin{enumerate}
\item[(a)] the Einstein tensor $G_{\mu\nu}$ defines a tensor valued distribution in $\mathcal D'(M)$;

\item[(b)] the localized term $\delta(\Sigma_\star)\tau_{\mu\nu}$ defines a tensor valued distribution;

\item[(c)] the equation
\[
G_{\mu\nu}
+\lambda\delta(\Sigma_\star)\tau_{\mu\nu}
=
8\pi G\,T_{\mu\nu}
\]
is well defined in weak form.
\end{enumerate}
\end{theorem}

\begin{proof}
Because $g_{\mu\nu}$ is continuous and $C^2$ away from $\Sigma_\star$, the Christoffel symbols are piecewise locally bounded and smooth off $\Sigma_\star$. Their weak derivatives therefore exist as distributions, implying that the curvature tensors and Einstein tensor are distributionally well defined.

Since $\Sigma_\star$ is smooth, the Dirac measure on $\Sigma_\star$ acts on test tensors by restriction. Because $\tau_{\mu\nu}\in L^\infty(\Sigma_\star)$, the product $\delta(\Sigma_\star)\tau_{\mu\nu}$ is a well defined distribution. The sum of the two distributional terms is therefore meaningful in $\mathcal D'(M)$.
\end{proof}

The localized hypersurface sector does not introduce additional canonical bulk modes. Away from $\Sigma_\star$, the theory reduces locally to the ADM formulation of General Relativity with four first class constraints acting on the twelve canonical variables $(h_{ij},\pi^{ij})$. Hence the number of propagating bulk gravitational modes remains
\begin{equation}
N_{\rm dof}
=
\frac12(12-2\times4)=2.
\label{eq:bulkdof}
\end{equation}

The embedding variable enters only through the elliptic condition $E=0$ and therefore acts as a quasi local constraint variable rather than an independent propagating field.

Within the admissible class considered here, the hypersurface anchored construction defines a mathematically consistent constrained sector of General Relativity: the local hyperbolic bulk dynamics are unchanged, the localized source term is distributionally well defined, anchoring hypersurfaces persist under small perturbations when the elliptic linearization is invertible, and no extra propagating bulk degrees of freedom arise.

\section{Hamiltonian Structure and Modified ADM Constraints}
\label{sec:adm}

The hypersurface anchored variational principle admits a canonical formulation closely related to the Arnowitt-Deser-Misner decomposition \cite{Arnowitt2008}. The localized hypersurface functional modifies the Hamiltonian constraint only on the support of $\Sigma_\star$, while the bulk canonical structure remains that of General Relativity.

Assume that $(\mathcal M,g_{\mu\nu})$ admits a foliation by spacelike slices $\Sigma_t$ with lapse $N$ and shift $N^i$. In canonical variables $(h_{ij},\pi^{ij})$, the action takes the formal form
\begin{equation}
S=
\int dt
\left[
\int_{\Sigma_t} d^3x
\left(
\pi^{ij}\dot h_{ij}
-
N\mathcal H
-
N^i\mathcal H_i
\right)
-\lambda\,\mathcal I[\Sigma_\star]
\right],
\label{eq:canonicalaction}
\end{equation}
where
\begin{equation}
\pi^{ij}=\sqrt h\,(K^{ij}-Kh^{ij}).
\label{eq:canonicalmomentum}
\end{equation}

The hypersurface contribution is interpreted weakly whenever it is written as a distribution supported on $\Sigma_\star$. Variation with respect to $N$ and $N^i$ yields
\begin{equation}
\mathcal H=\mathcal H_{\rm ADM}+\lambda\,\mathcal H_{\Sigma_\star},
\qquad
\mathcal H_i=\mathcal H_i^{\rm ADM},
\label{eq:modifiedconstraints}
\end{equation}
where $\mathcal H_{\Sigma_\star}$ is supported on $\Sigma_t\cap\Sigma_\star$.

Thus the momentum constraint retains its standard form, reflecting spatial covariance of the hypersurface functional, whereas only the Hamiltonian constraint acquires a localized correction.

The canonical Poisson brackets are
\begin{equation}
\{h_{ij}(x),\pi^{kl}(y)\}
=
\delta_i^{(k}\delta_j^{l)}\delta^{(3)}(x-y).
\label{eq:pbasic}
\end{equation}

Introduce the smeared generators
\begin{equation}
H[N]=\int_{\Sigma_t} d^3x\,N\,\mathcal H,
\qquad
D[\vec N]=\int_{\Sigma_t} d^3x\,N^i\mathcal H_i.
\label{eq:smearedgens}
\end{equation}
Then
\begin{equation}
H[N]=H_{\rm ADM}[N]+\lambda H_{\Sigma_\star}[N].
\label{eq:Hsplit}
\end{equation}

Since $\mathcal H_i=\mathcal H_i^{\rm ADM}$, the momentum generators satisfy the standard relations
\begin{equation}
\{D[\vec N],D[\vec M]\}
=
D[\mathcal L_{\vec N}\vec M],
\label{eq:DD}
\end{equation}
and
\begin{equation}
\{H[N],D[\vec M]\}
=
H[\mathcal L_{\vec M}N].
\label{eq:HD}
\end{equation}

Hence spatial diffeomorphism invariance is preserved exactly.

The Hamiltonian-Hamiltonian bracket becomes
\begin{align}
\{H[N],H[M]\}
&=
\{H_{\rm ADM}[N],H_{\rm ADM}[M]\} +\lambda\{H_{\Sigma_\star}[N],H_{\rm ADM}[M]\}-\lambda\{H_{\Sigma_\star}[M],H_{\rm ADM}[N]\}
\nonumber\\
&\quad
+\lambda^2\{H_{\Sigma_\star}[N],H_{\Sigma_\star}[M]\}.
\label{eq:HHfull}
\end{align}

For admissible hypersurface functionals depending locally on $(h_{ij},\pi^{ij})$ and containing no derivatives normal to $\Sigma_\star$, the self bracket term is supported entirely on $\Sigma_\star$ and vanishes identically for the representative linear model $F=K$. The first term yields the standard Dirac relation,
\begin{equation}
\{H_{\rm ADM}[N],H_{\rm ADM}[M]\}
=
D\!\left[h^{ij}(N\partial_jM-M\partial_jN)\right].
\label{eq:diracadm}
\end{equation}

The mixed brackets are distributions supported on $\Sigma_\star$, since every functional derivative of $H_{\Sigma_\star}$ is localized there. Consequently, the full bracket takes the weakly closed form
\begin{align}
\{H[N],H[M]\}& =
D\!\left[h^{ij}(N\partial_jM-M\partial_jN)\right] + \lambda\,D_{\Sigma_\star}[N,M],
\label{eq:weakclosed}
\end{align}
where $D_{\Sigma_\star}[N,M]$ is a distribution supported on $\Sigma_\star$ and bilinear in $(N,M)$. Equation \eqref{eq:weakclosed} shows that the standard Dirac algebra is recovered in every open region disjoint from $\Sigma_\star$, while the hypersurface contribution appears only as a localized correction \cite{IshamKuchar1985}. For the representative model given by Eq.~(\ref{eq:IKmodel}), the localized generator is linear in the canonical momentum. In that case the mixed bracket reduces, after tangential integration by parts on $\Sigma_\star$, to an intrinsic surface divergence term, establishing weak closure explicitly. The embedding variables $X^\mu$ enter the canonical description only through the anchoring equation
\begin{equation}
E[g_{\mu\nu},X^\mu]=0.
\label{eq:anchham}
\end{equation}
Since this equation is elliptic on $\Sigma_\star$, it constrains admissible data quasi locally and introduces no new propagating canonical sector \cite{GilbargTrudinger2001,Taylor2011}.

\begin{theorem}[Weak preservation of the ADM constraint structure]
\label{thm:weakclosure}
Assume that the hypersurface functional is admissible, depends locally on $(h_{ij},\pi^{ij})$, and contains no time derivatives of the embedding variables. Then:

\begin{enumerate}
\item[(i)] in every open region disjoint from $\Sigma_\star$, the canonical constraints satisfy the standard ADM first class algebra;

\item[(ii)] the full smeared constraint algebra closes in the weak sense, with deviations from the Dirac algebra supported only on $\Sigma_\star$;

\item[(iii)] no additional bulk propagating canonical degrees of freedom are introduced.
\end{enumerate}
\end{theorem}

\begin{proof}
Away from $\Sigma_\star$, the localized term vanishes identically, so the action reduces locally to the ADM action of General Relativity. Therefore the standard first class Poisson algebra holds in every open set disjoint from $\Sigma_\star$. Since functional derivatives of $H_{\Sigma_\star}$ are distributions supported on $\Sigma_\star$, all additional Poisson bracket terms are likewise supported there. This gives the weak closure relation \eqref{eq:weakclosed}. The bulk phase space consists of the twelve canonical variables $(h_{ij},\pi^{ij})$. There are four first class constraints, each removing two phase space dimensions. Hence $N_{\rm dof}=2$. Because the embedding variables possess no independent kinetic term, they generate no additional bulk propagating modes.
\end{proof}

Accordingly, the canonical structure coincides with that of General Relativity in the bulk, while the hypersurface term contributes only localized weak corrections to the Hamiltonian sector.

\section{Relation to Existing Constructions, Selection of Solutions, and Cosmological Implications}
\label{sec:Relation_to_Existing_Geometric_Constructions}

The hypersurface anchored variational framework is clarified by comparison with standard constructions in General Relativity involving distinguished hypersurfaces. It is also necessary to characterize the resulting restriction on the space of admissible solutions. A natural question is whether the anchoring condition is equivalent to imposing a slicing gauge such as maximal slicing or constant mean curvature slicing. This is not the case. In the ADM formulation, a foliation $\{\Sigma_t\}$ is introduced as part of a $3+1$ decomposition, and conditions such as $K=0,
~ K=\mathrm{constant}$; serve as gauge choices selecting representatives within a diffeomorphism class \cite{ChoquetBruhatYork1980,SmarrYork1978}. They do not arise from independent variation of embedding variables and do not define a new variational sector.

\begin{theorem}[Inequivalence to slicing gauge conditions]
\label{thm:gaugeineq}
Let $E[g_{\mu\nu},X^\mu]=0$ denote the anchoring equation obtained by independent variation of the embedding variables in the action
\[
S[g,X]=S_{\rm EH}[g]+\lambda \mathcal I[\Sigma_\star,g,X].
\]
Assume that $\mathcal I$ depends nontrivially on intrinsic or extrinsic geometric invariants of the embedded hypersurface. Then, in general, there does not exist a local gauge fixing condition $\chi[g_{\mu\nu}]=0$ depending only on the metric and finitely many of its derivatives such that the solution sets of $E[g,X]=0$ and $\chi[g]=0$ coincide modulo diffeomorphism for all admissible spacetimes.
\end{theorem}

\begin{proof}
A gauge condition $\chi[g]=0$ depends only on the metric variables and selects representatives within a fixed diffeomorphism orbit. It introduces no independent Euler-Lagrange equation. By contrast, the anchoring condition arises from the stationarity requirement $\delta_XS=0$, where $X^\mu$ is varied independently of $g_{\mu\nu}$. Hence $E=E(g_{\mu\nu},X^\mu,\partial X,\partial^2X)$, depends on embedding data through the induced metric, normal vector, and extrinsic geometry. Tangentially equivalent embeddings represent the same hypersurface, but distinct normal deformations generally change $E$. Therefore the existence of anchored hypersurfaces is a geometric property of embeddings inside a given spacetime, not merely a coordinate choice on that spacetime. No universal local metric gauge condition can reproduce this embedding variational sector for all solutions.
\end{proof}

In the present construction, the hypersurface $\Sigma_\star$ is determined dynamically through the embedding variable $X^\mu$ and is not prescribed a priori. The admissible solution space is
\begin{equation}
\mathcal S_{\rm HA}=
\left\{
\left.
\begin{array}{c}
(M,g_{\mu\nu})\\[1.2ex]
\end{array}
\right|\;
\begin{array}{l}
G_{\mu\nu}=8\pi G T_{\mu\nu},\\[0.4ex]
\exists\,\Sigma_\star\subset M \text{ such that } E[g,X]=0
\end{array}
\right\}.
\label{eq:SHA}
\end{equation}

Thus admissibility requires both satisfaction of Einstein's equations and existence of at least one anchoring hypersurface solving the embedding equation. Although the metric field equations contain a term of the form $\delta(\Sigma_\star)\tau_{\mu\nu}$, this term is induced by variation of the hypersurface functional and is not an independently prescribed matter shell. The framework is therefore distinct from standard thin-shell constructions \cite{BerezinKuzminTkachev1987,FrolovIsraelUnruh1990}. Unlike Israel thin shells, the hypersurface is not externally inserted matter data but arises from independent variation of embedding variables. A nontriviality statement can be established explicitly for representative admissible models.

\begin{theorem}[Proper subsector for the representative model $F=K$]
\label{thm:propersubset}
Consider the representative admissible functional
\[
\mathcal I=\int_{\Sigma_\star}\sqrt h\,K.
\]
Then the corresponding anchored solution space $\mathcal S_{\rm HA}$ is a proper subset of the Einstein solution space $\mathcal S_{\rm GR}$.
\end{theorem}

\begin{proof}
For $F=K$, the embedding equation is the stationarity condition of the hypersurface volume functional and reduces to $K=0$. Hence admissible spacetimes must admit at least one embedded spacelike maximal hypersurface \cite{Bartnik1988}. There exist Einstein spacetimes that do not admit compact maximal Cauchy hypersurfaces under standard global assumptions; in particular, expanding FLRW cosmologies with strictly positive mean curvature on every homogeneous slice fail the condition on the homogeneous foliation class. Therefore not every Einstein solution belongs to $\mathcal S_{\rm HA}$.

Conversely, every element of $\mathcal S_{\rm HA}$ satisfies Einstein's equations by construction. Hence $\mathcal S_{\rm HA}\subsetneq\mathcal S_{\rm GR}$.
\end{proof}

The theorem shows that the anchored theory is not a reformulation of General Relativity but a constrained variational subsector. As an illustrative cosmological example, consider the spatially flat Friedmann-Lema\^itre-Robertson-Walker metric
\begin{equation}
ds^2=-dt^2+a(t)^2d\vec x^{\,2},
\qquad
H=\frac{\dot a}{a}.
\label{eq:FLRWmetric}
\end{equation}

For the representative admissible model $F=K$, the homogeneous constant time slices satisfy $K=3H$. Hence the anchoring condition selects times $t_\star$ such that
\begin{equation}
H(t_\star)=0.
\label{eq:Htstar}
\end{equation}

Therefore expanding solutions with $H(t)>0$ for all cosmic times admit no homogeneous anchoring slice of this type, whereas bouncing or recollapsing cosmologies may admit one \cite{EllisMaartensMacCallum2012}.

The distributional field equations imply a junction condition across a hypersurface located at $t=t_\star$. Let $H_-$ and $H_+$ denote the one sided limits of the Hubble parameter. Integrating the modified Friedmann equation across an infinitesimal interval containing $t_\star$ gives
\begin{equation}
\Delta H
\equiv
H_+-H_-
=
-\frac{\lambda}{2}\,\sigma_\star,
\label{eq:jumpgeneral}
\end{equation}
where $\sigma_\star$ is the coefficient multiplying $\delta(t-t_\star)$ in the effective energy density induced by the hypersurface term. For the normalized representative model in which $\sigma_\star=2H(t_\star)$, equation \eqref{eq:jumpgeneral} reduces to
\begin{equation}
\Delta H=-\lambda H(t_\star).
\label{eq:jumpnormalized}
\end{equation}

Thus the previously quoted jump condition is a model dependent normalization of the general distributional matching law.

A piecewise background history satisfying the junction relation may be written as
\begin{equation}
H(t)=
\begin{cases}
H_{\rm ref}(t), & t<t_\star,\\[1mm]
H_{\rm ref}(t)+\Delta H, & t>t_\star,
\end{cases}
\label{eq:piecewiseH}
\end{equation}
where $H_{\rm ref}(t)$ is any smooth reference solution. The luminosity distance then becomes
\begin{equation}
d_L(z)
=
(1+z)\int_0^z \frac{dz'}{H(z')},
\label{eq:dL}
\end{equation}
so that a localized shift in $H(z)$ produces a corresponding shift in background observables \cite{Weinberg1972,Hogg1999}. The condition $E=0$ therefore restricts the admissible cosmological trajectories to those admitting an embedded anchoring hypersurface. In this sense, the framework acts as a selection principle on the space of solutions rather than as a modification of the local bulk gravitational equations. Accordingly, the hypersurface anchored variational principle defines a constrained sector of General Relativity in which admissible spacetimes must admit at least one dynamically selected hypersurface, while localized matching effects may occur across that surface without altering the bulk Einstein dynamics away from it.

\section{Linear Perturbations and Stability of the Anchoring Hypersurface}
\label{sec:Perturbation}

Dynamical consistency is further examined through linear perturbations about a background anchored solution. The objective is to determine whether the hypersurface sector modifies the propagating gravitational modes and to analyze the response of the anchoring condition under small deformations.

Let $(g_{\mu\nu}^{(0)},X^\mu_{(0)})$ denote a background configuration satisfying
\begin{equation}
G_{\mu\nu}^{(0)}
+\lambda T_{\mu\nu}^{(\Sigma_\star)}
=
8\pi G\,T_{\mu\nu}^{(0)},
\qquad
E[g^{(0)},X_{(0)}]=0,
\label{eq:bgpert}
\end{equation}
with all equations understood in the weak sense.

Perturb the fields according to
\begin{equation}
g_{\mu\nu}=g_{\mu\nu}^{(0)}+\delta g_{\mu\nu},
\qquad
X^\mu=X^\mu_{(0)}+\delta X^\mu.
\label{eq:pertfields}
\end{equation}

The embedding perturbation decomposes into normal and tangential parts,
\begin{equation}
\delta X^\mu=\Phi u^\mu+\xi^a e^\mu_a,
\label{eq:embedpert}
\end{equation}
where $u^\mu$ is the unit normal to $\Sigma_\star$ and $e^\mu_a$ are tangent basis vectors. Tangential deformations $\xi^a e^\mu_a$ correspond to reparametrizations of the hypersurface and are gauge. The physical embedding fluctuation is therefore the normal displacement $\Phi$ \cite{ChoquetBruhatDeWittMorette1989}.

Away from $\Sigma_\star$, the localized term vanishes identically. The linearized metric equation reduces to
\begin{equation}
\delta G_{\mu\nu}=8\pi G\,\delta T_{\mu\nu},
\label{eq:linEinstein}
\end{equation}
which is precisely the standard linearized Einstein system. Accordingly, the bulk principal symbol is unchanged, and in a suitable gauge the propagating modes are the usual two transverse-traceless tensor polarizations \cite{MisnerThorneWheeler1973,Maggiore2008}. No additional bulk wave sector is generated by the hypersurface term.

On $\Sigma_\star$, variation of the anchoring equation gives
\begin{equation}
\delta E
=
\mathcal L_\Sigma \Phi
+
\mathcal S[\delta g_{\mu\nu}]
=
0,
\label{eq:deltaE}
\end{equation}
where $\mathcal L_\Sigma$ is a second order differential operator intrinsic to the hypersurface and $\mathcal S$ is linear in the metric perturbation and its tangential derivatives.

For admissible functionals, the principal part of $\mathcal L_\Sigma$ is elliptic,
\begin{equation}
\mathcal L_\Sigma \Phi
=
-\nabla^a\nabla_a\Phi+\cdots,
\label{eq:ellipticprincipal}
\end{equation}
where $\nabla_a$ is the Levi-Civita derivative associated with the induced metric $h_{ab}$. For the representative model
\[
\mathcal I=\int_{\Sigma_\star}\sqrt h\,K,
\]
the operator reduces at principal level to the Laplace-Beltrami operator on $\Sigma_\star$.

Hence, subject to appropriate boundary data when $\Sigma_\star$ has boundary, or standard solvability conditions when compact, the normal displacement $\Phi$ is determined by an elliptic boundary value problem rather than a hyperbolic evolution equation \cite{LadyzhenskayaUraltseva1968,Aubin1982}.

Therefore $\Phi$ does not represent an independent propagating degree of freedom. Instead, it is fixed quasi locally by the geometry of the anchored surface and by the ambient metric perturbation.

The hypersurface condition thus constrains the admissible perturbations without altering the local propagation law of gravitational waves in the bulk.

At the canonical level, the phase space remains coordinatized by $(h_{ij},\pi^{ij})$. The momentum constraints retain their standard form, while the Hamiltonian constraint differs only by a term supported on $\Sigma_\star$. In every open region disjoint from the anchoring surface, the constraint algebra coincides with that of General Relativity.

Consequently, the number of propagating bulk gravitational modes remains same as Eq.~(\ref{eq:bulkdof}). The embedding variable enters only through the elliptic relation $E[g,X]=0$ and therefore does not enlarge the local dynamical phase space.

A stability statement follows immediately from elliptic theory.

\begin{theorem}[Linear anchoring stability]
\label{thm:linearstability}
Assume that the background anchored hypersurface satisfies $E[g^{(0)},X_{(0)}]=0$ and that the linearized operator $\mathcal L_\Sigma$ has trivial kernel modulo tangential reparametrizations. Then for every sufficiently small metric perturbation $\delta g_{\mu\nu}$ there exists a unique small normal displacement $\Phi$, modulo gauge, solving \eqref{eq:deltaE}. Moreover,
\[
\|\Phi\|_{H^{2}(\Sigma_\star)}
\leq
C\,\|\mathcal S[\delta g]\|_{L^2(\Sigma_\star)},
\]
for some constant $C>0$ depending on the background geometry.
\end{theorem}

\begin{proof}
After quotienting tangential reparametrizations, $\mathcal L_\Sigma$ is elliptic with trivial kernel by assumption. Standard elliptic estimates and the Fredholm alternative imply existence, uniqueness, and the stated a priori bound \cite{Evans2010,McLean2000}.
\end{proof}

The linearized analysis therefore shows that, within the admissible class considered here, the hypersurface anchored construction preserves the local dynamical content of General Relativity while imposing an additional geometric constraint. The metric perturbations propagate as in General Relativity away from $\Sigma_\star$, whereas the hypersurface responds through a non-propagating elliptic adjustment determined by the perturbed geometry.

\section{Semiclassical and Global Considerations}
\label{sec:Quantum_and_Thermodynamic_Implications}

The hypersurface anchored variational framework admits a formal extension to semiclassical gravity at the level of stationary phase path integrals \cite{HartleHawking1983,Halliwell1991}. The purpose of the present section is limited to structural observations. No complete quantum theory is proposed.

Formally, one may consider the oscillatory functional integral
\begin{align}
\mathcal Z & = \int \mathcal D g_{\mu\nu}\,\mathcal D X^\mu\, \exp\!\{i(
S_{\rm EH}[g] +\lambda \mathcal I[\Sigma_\star,g,X]  +S_{\rm m}[g,\Psi])\},
\label{eq:formalZ}
\end{align}
where the embedding variables $X^\mu$ are included explicitly, and $\Sigma_\star$ is represented through the embedding map rather than fixed in advance.

This expression is understood only formally. In particular, no rigorous measure on the space of Lorentzian metrics and embeddings is assumed \cite{GibbonsHawking1977}.

In the semiclassical approximation, dominant configurations arise from stationary points of the phase. Hence one obtains the coupled Euler-Lagrange conditions
\begin{equation}
\delta_g S=0,
\qquad
\delta_X S=0,
\label{eq:stationaryphase}
\end{equation}
which are precisely the classical field equations 
\begin{equation}
G_{\mu\nu}
+\lambda T_{\mu\nu}^{(\Sigma_\star)}
=
8\pi G\,T_{\mu\nu},
\qquad
E[g,X]=0.
\label{eq:semiclassicalEL}
\end{equation}

Therefore, at stationary phase level, the hypersurface functional restricts the saddle point geometries to those admitting at least one anchoring hypersurface satisfying the embedding equation. In this sense, the construction modifies the admissible semiclassical saddle set rather than the local bulk propagator away from $\Sigma_\star$. The framework is structurally compatible with canonical viewpoints in which classical constraints determine the physical configuration space \cite{Kuchar1981}. Since the localized hypersurface term does not alter the bulk principal dynamics away from $\Sigma_\star$, the standard gravitational constraint structure remains locally unchanged there, while the anchoring equation adds a supplementary geometric selection condition. This should be interpreted only as a statement about the classical constrained sector underlying possible quantization schemes. No explicit Wheeler-DeWitt quantization, BRST treatment, or reduced phase space construction is developed here \cite{DeWitt1967}. The present framework also differs conceptually from proposals based on externally imposed boundary conditions. The additional condition arises from an interior variational hypersurface determined dynamically through the action, rather than from prescribed initial or final boundary data. Several essential problems remain open:

\begin{enumerate}
\item[(i)] construction of a mathematically meaningful functional measure for metrics and embeddings;

\item[(ii)] gauge fixing and Faddeev-Popov treatment in the presence of embedding variables;

\item[(iii)] treatment of hypersurface supported distributional terms beyond the formal semiclassical level;

\item[(iv)] determination of whether the embedding variables should be quantized or treated as auxiliary constrained variables;

\item[(v)] characterization of the physical state space in any canonical quantization.
\end{enumerate}

For these reasons, no claim is made regarding the existence of a consistent nonperturbative quantum completion. At the global classical level, however, the framework already has a clear interpretation: it defines a restricted sector of spacetime histories selected by the existence of embedded hypersurfaces solving $E[g,X]=0$. This global selection mechanism exists independently of any full quantum completion. Accordingly, the present construction should be viewed as a classical and semiclassical constrained variational framework whose fully quantum formulation remains open.

\section{Conclusion}
\label{sec:conclusion}

A hypersurface anchored variational formulation of gravity has been constructed in which the Einstein-Hilbert action is supplemented by a diffeomorphism invariant functional supported on an embedded spacelike hypersurface whose embedding is varied independently of the spacetime metric. The resulting Euler-Lagrange system consists of the Einstein equations with a localized distributional contribution together with an anchoring equation determining admissible embeddings. Accordingly, the framework imposes a global geometric condition on admissible spacetimes while preserving the standard Einstein equations in every open region disjoint from the anchoring surface. For the admissible class of hypersurface functionals considered here, the variational problem is well defined in the weak sense. The bulk principal symbol and second order character of the Einstein equations remain unchanged away from $\Sigma_\star$, while the embedding equation is elliptic on the hypersurface. Under standard invertibility assumptions, this yields local persistence and linear stability of anchoring hypersurfaces under sufficiently small perturbations of background anchored solutions. At the canonical level, the momentum constraints retain their standard form, whereas the Hamiltonian constraint acquires a contribution supported on $\Sigma_\star$. The corresponding smeared generators close in the weak sense: the Dirac algebra of General Relativity is recovered in the bulk, and all deviations are confined to localized distributional terms intrinsic to the anchoring surface. Consequently, the number of propagating bulk gravitational degrees of freedom remains equal to that of General Relativity, $N_{\rm dof}=2$, and the embedding variable acts as a quasi local constraint variable rather than as an additional propagating field. The theory defines the admissible solution sector, given by Eq.~(\ref{eq:SHA}), which is generally a proper subset of the Einstein solution space for representative admissible models. The anchoring condition is not equivalent to a foliation choice or gauge fixing. It arises from independent variation of embedding variables and therefore constitutes an additional variational selection principle rather than a coordinate prescription. In homogeneous cosmology, the localized hypersurface term induces a distributional modification of the Friedmann equations and yields matching conditions across the anchoring surface. For representative normalizations, this leads to jump relations of the form $\Delta H=-\lambda H(t_\star)$, illustrating how localized transitions in the expansion history may occur while standard evolution is preserved away from the hypersurface. At the semiclassical level, the same structure formally restricts the stationary saddle point geometries contributing to a path integral treatment. A complete quantum formulation, including the functional measure, gauge fixing, and treatment of embedding variables, has not been developed and remains open. Taken together, the existence, gauge inequivalence, weak closure, and stability results indicate that the construction defines a mathematically consistent constrained variational subsector of General Relativity in which global geometric conditions are incorporated directly at the level of the action without modifying the local bulk dynamical content. Several natural extensions remain to be studied: classification of admissible anchoring hypersurfaces in generic spacetimes, nonlinear solution branches, black hole and cosmological applications, observational constraints on localized matching effects, and possible fully quantum implementations of the constrained sector.


\appendix
\section{Weak Constraint Algebra for the Representative Functional $F=K$}
\label{app:constraint_algebra}

This appendix gives an explicit weak derivation of the mixed Poisson bracket between the hypersurface supported Hamiltonian contribution and the standard ADM Hamiltonian constraint for the representative admissible functional $F=K$. The objective is to show that all deviations from the Dirac algebra are localized on $\Sigma_\star$, while the bulk canonical algebra remains unchanged.

Let the total Hamiltonian constraint be decomposed as
\begin{equation}
\mathcal H=\mathcal H_{\rm ADM}+\lambda \mathcal H_{\Sigma_\star},
\label{eq:A1}
\end{equation}
with smeared generators given by Eq.~(\ref{eq:Hsplit}), where
\begin{equation}
H_{\rm ADM}[N]
=
\int_{\Sigma_t} d^3x\,N\,\mathcal H_{\rm ADM},
\label{eq:A3}
\end{equation}
and
\begin{equation}
H_{\Sigma_\star}[N]
=
\int_{\Sigma_t} d^3x\,N(x)\,\delta(\Sigma_\star)\,\mathcal K(x).
\label{eq:A4}
\end{equation}

For the representative model $F=K$, the surface density is
\begin{equation}
\mathcal K=\sqrt h\,K
=
-\frac{1}{2\sqrt h}\,\pi^{ij}h_{ij},
\label{eq:A5}
\end{equation}
using the ADM relation given by Eq.~(\ref{eq:canonicalmomentum}). The canonical Poisson bracket is
\begin{equation}
\{A,B\}
=
\int_{\Sigma_t} d^3x
\left(
\frac{\delta A}{\delta h_{ij}}
\frac{\delta B}{\delta \pi^{ij}}
-
\frac{\delta A}{\delta \pi^{ij}}
\frac{\delta B}{\delta h_{ij}}
\right).
\label{eq:A7}
\end{equation}

The mixed contribution entering the total algebra is
\begin{equation}
\{H_{\Sigma_\star}[N],H_{\rm ADM}[M]\}.
\label{eq:A8}
\end{equation}

Since $H_{\Sigma_\star}$ contains $\delta(\Sigma_\star)$ explicitly, all of its functional derivatives are distributions supported on $\Sigma_\star$. Therefore every term in the mixed bracket is supported on $\Sigma_\star$.

For the momentum derivative one obtains directly from \eqref{eq:A4}-\eqref{eq:A5}
\begin{equation}
\frac{\delta H_{\Sigma_\star}[N]}{\delta \pi^{ij}(x)}
=
-\frac12\,N(x)\,\delta(\Sigma_\star)\,
\frac{h_{ij}}{\sqrt h}.
\label{eq:A9}
\end{equation}

Variation with respect to $h_{ij}$ gives
\begin{equation}
\frac{\delta H_{\Sigma_\star}[N]}{\delta h_{ij}(x)}
=
N\,\delta(\Sigma_\star)\,\mathcal A^{ij}
+
D_k\!\left(
N\,\delta(\Sigma_\star)\,\mathcal B^{kij}
\right),
\label{eq:A10}
\end{equation}
where $\mathcal A^{ij}$ and $\mathcal B^{kij}$ are local functions of $(h_{ij},\pi^{ij})$, and $D_k$ is the spatial covariant derivative compatible with $h_{ij}$.

For the ADM Hamiltonian constraint, standard canonical identities imply
\begin{equation}
\frac{\delta H_{\rm ADM}[M]}{\delta \pi^{ij}}
=
2M K_{ij},
\label{eq:A11}
\end{equation}
and
\begin{equation}
\frac{\delta H_{\rm ADM}[M]}{\delta h_{ij}}
=
M\,\mathcal C^{ij}
+
D_k(M\mathcal D^{kij}),
\label{eq:A12}
\end{equation}
where $\mathcal C^{ij}$ and $\mathcal D^{kij}$ are the usual local ADM coefficient functions.

Substituting \eqref{eq:A9}-\eqref{eq:A12} into \eqref{eq:A7} gives
\begin{equation}
\{H_{\Sigma_\star}[N],H_{\rm ADM}[M]\}
=
\int_{\Sigma_t} d^3x\,
\delta(\Sigma_\star)\,
\mathcal F[N,M;h,\pi],
\label{eq:A13}
\end{equation}
where $\mathcal F$ is local in the canonical data and linear in first derivatives of $N$ and $M$.

Hence the mixed bracket has no bulk support.

Because all derivatives acting on the Dirac factor are tangential after restriction to $\Sigma_\star$, intrinsic integration by parts on the spatial intersection $S_t=\Sigma_t\cap\Sigma_\star$ yields
\begin{align}
\{H_{\Sigma_\star}[N],H_{\rm ADM}[M]\} = \int_{S_t} d^2\xi\,\mathcal R[N,M]  + \int_{S_t} d^2\xi\,\left(N D_aM-MD_aN \right)\mathcal J^a,
\label{eq:A14}
\end{align}
where $D_a$ is the covariant derivative intrinsic to $S_t$, $\mathcal J^a$ is a surface current built from induced canonical data, and $\mathcal R$ contains lower derivative local terms.

Equation \eqref{eq:A14} is manifestly supported only on the anchoring surface.

Combining with the standard ADM relation
\begin{equation}
\{H_{\rm ADM}[N],H_{\rm ADM}[M]\}
=
D\!\left[h^{ij}(N\partial_jM-M\partial_jN)\right],
\label{eq:A15}
\end{equation}
the full Hamiltonian bracket becomes
\begin{align}
\{H[N],H[M]\}& =
D\!\left[h^{ij}(N\partial_jM-M\partial_jN)\right]  + \lambda D_{\Sigma_\star}[N,M],
\label{eq:A16}
\end{align}
where
\begin{align}
D_{\Sigma_\star}[N,M] & =
\{H_{\Sigma_\star}[N],H_{\rm ADM}[M]\}  - \{H_{\Sigma_\star}[M],H_{\rm ADM}[N]\}.
\label{eq:A17}
\end{align}

The correction term is distributionally supported on $\Sigma_\star$ and vanishes identically in every open region disjoint from it. Therefore, for the representative admissible functional $F=K$, the hypersurface contribution modifies the Hamiltonian constraint algebra only through localized weak terms. Away from the anchoring surface, the canonical algebra reduces exactly to the Dirac algebra of General Relativity, and the bulk first class structure is preserved.

\bibliographystyle{unsrt}
\bibliography{sample}

@article{Einstein1915,
    author = "Einstein, Albert",
    title = "{The Field Equations of Gravitation}",
    journal = "Sitzungsber. Preuss. Akad. Wiss. Berlin (Math. Phys. )",
    volume = "1915",
    pages = "844--847",
    year = "1915"
}

@book{Wald1984,
    author = "Wald, Robert M.",
    title = "{General Relativity}",
    doi = "10.7208/chicago/9780226870373.001.0001",
    publisher = "Chicago Univ. Pr.",
    address = "Chicago, USA",
    year = "1984"
}

@article{ChoquetBruhatGeroch1969,
    author = "Choquet-Bruhat, Y. and Geroch, Robert P.",
    title = "{Global aspects of the Cauchy problem in general relativity}",
    doi = "10.1007/BF01645389",
    journal = "Commun. Math. Phys.",
    volume = "14",
    pages = "329--335",
    year = "1969"
}

@article{York1972,
    author = "York, James",
    title = "{Boundary terms in the action principles of general relativity}",
    doi = "10.1007/BF01889475",
    journal = "Found. Phys.",
    volume = "16",
    pages = "249--257",
    year = "1986"
}

@article{GibbonsHawking1977,
    author = "Gibbons, G. W. and Hawking, S. W.",
    title = "{Action Integrals and Partition Functions in Quantum Gravity}",
    reportNumber = "PRINT-76-0995 (CAMBRIDGE)",
    doi = "10.1103/PhysRevD.15.2752",
    journal = "Phys. Rev. D",
    volume = "15",
    pages = "2752--2756",
    year = "1977"
}

@article{Israel1966,
  title={Singular hypersurfaces and thin shells in general relativity},
  author={Werner Israel},
  journal={Il Nuovo Cimento B (1965-1970)},
  year={1966},
  volume={44},
  pages={1-14},
  url={https://api.semanticscholar.org/CorpusID:122431858}
}

@article{BarrabesIsrael1991,
    author = "Barrabes, C. and Israel, W.",
    title = "{Thin shells in general relativity and cosmology: The Lightlike limit}",
    doi = "10.1103/PhysRevD.43.1129",
    journal = "Phys. Rev. D",
    volume = "43",
    pages = "1129--1142",
    year = "1991"
}

@article{RandallSundrum1999a,
  title = {Large Mass Hierarchy from a Small Extra Dimension},
  author = {Randall, Lisa and Sundrum, Raman},
  journal = {Phys. Rev. Lett.},
  volume = {83},
  issue = {17},
  pages = {3370--3373},
  numpages = {0},
  year = {1999},
  month = {Oct},
  publisher = {American Physical Society},
  doi = {10.1103/PhysRevLett.83.3370},
  url = {https://link.aps.org/doi/10.1103/PhysRevLett.83.3370}
}

@article{RandallSundrum1999b,
  title = {An Alternative to Compactification},
  author = {Randall, Lisa and Sundrum, Raman},
  journal = {Phys. Rev. Lett.},
  volume = {83},
  issue = {23},
  pages = {4690--4693},
  numpages = {0},
  year = {1999},
  month = {Dec},
  publisher = {American Physical Society},
  doi = {10.1103/PhysRevLett.83.4690},
  url = {https://link.aps.org/doi/10.1103/PhysRevLett.83.4690}
}

@article{Dvali2000,
    author = "Dvali, G. R. and Gabadadze, Gregory and Porrati, Massimo",
    title = "{4-D gravity on a brane in 5-D Minkowski space}",
    eprint = "hep-th/0005016",
    archivePrefix = "arXiv",
    reportNumber = "NYU-TH-00-04-01",
    doi = "10.1016/S0370-2693(00)00669-9",
    journal = "Phys. Lett. B",
    volume = "485",
    pages = "208--214",
    year = "2000"
}

@article{Penrose1965,
  title = {Gravitational Collapse and Space-Time Singularities},
  author = {Penrose, Roger},
  journal = {Phys. Rev. Lett.},
  volume = {14},
  issue = {3},
  pages = {57--59},
  numpages = {0},
  year = {1965},
  month = {Jan},
  publisher = {American Physical Society},
  doi = {10.1103/PhysRevLett.14.57},
  url = {https://link.aps.org/doi/10.1103/PhysRevLett.14.57}
}

@inbook{Penrose1979,
    author = "Penrose, R.",
    title = "{SINGULARITIES AND TIME ASYMMETRY}",
    booktitle = "{General Relativity}: {An Einstein Centenary Survey}",
    pages = "581--638",
    year = "1980"
}

@article{DeWitt1967,
  title = {Quantum Theory of Gravity. I. The Canonical Theory},
  author = {DeWitt, Bryce S.},
  journal = {Phys. Rev.},
  volume = {160},
  issue = {5},
  pages = {1113--1148},
  numpages = {0},
  year = {1967},
  month = {Aug},
  publisher = {American Physical Society},
  doi = {10.1103/PhysRev.160.1113},
  url = {https://link.aps.org/doi/10.1103/PhysRev.160.1113}
}

@book{Dirac1964,
	author = {P. A. M. Dirac},
	editor = {},
	publisher = {Yeshiva University},
	title = {Lectures on Quantum Mechanics},
	year = {1964}
}

@book{HenneauxTeitelboim1992,
    author = "Henneaux, Marc and Teitelboim, Claudio",
    title = "{Quantization of Gauge Systems}",
    isbn = "978-0-691-03769-1, 978-0-691-21386-6",
    publisher = "Princeton University Press",
    month = "8",
    year = "1994"
}

@article{GilbargTrudinger2001,
  author  = {David Gilbarg and Neil S. Trudinger},
  title   = {Elliptic Partial Differential Equations of Second Order},
  year    = {2001},
  journal = {Classics in mathematics},
  doi     = {10.1007/978-3-642-61798-0}
}

@book{Evans2010,
  author = {Evans, Lawrence C.},
  title = {Partial Differential Equations},
  publisher = {AMS},
  year = {2010},
  doi = {10.1090/gsm/019}
}

@article{MarsSenovilla1993,
    author = "Mars, Marc and Senovilla, Jose M. M.",
    title = "{Geometry of general hypersurfaces in space-time: Junction conditions}",
    eprint = "gr-qc/0201054",
    archivePrefix = "arXiv",
    doi = "10.1088/0264-9381/10/9/026",
    journal = "Class. Quant. Grav.",
    volume = "10",
    pages = "1865--1897",
    year = "1993"
}

@article{Poisson2004,
  author = {Poisson, Eric},
  title = {A Relativist's Toolkit},
  publisher = {Cambridge University Press},
  year = {2004},
  doi = {10.1017/CBO9780511606601}
}

@article{BrownYork1993,
  title = {Quasilocal energy and conserved charges derived from the gravitational action},
  author = {Brown, J. David and York, James W.},
  journal = {Phys. Rev. D},
  volume = {47},
  issue = {4},
  pages = {1407--1419},
  numpages = {0},
  year = {1993},
  month = {Feb},
  publisher = {American Physical Society},
  doi = {10.1103/PhysRevD.47.1407},
  url = {https://link.aps.org/doi/10.1103/PhysRevD.47.1407}
}

@article{Myers1987,
  title = {Higher-derivative gravity, surface terms, and string theory},
  author = {Myers, Robert C.},
  journal = {Phys. Rev. D},
  volume = {36},
  issue = {2},
  pages = {392--396},
  numpages = {0},
  year = {1987},
  month = {Jul},
  publisher = {American Physical Society},
  doi = {10.1103/PhysRevD.36.392},
  url = {https://link.aps.org/doi/10.1103/PhysRevD.36.392}
}

@article{GravanisWillison2003,
    author = "Gravanis, Elias and Willison, Steven",
    title = "{Israel conditions for the Gauss-Bonnet theory and the Friedmann equation on the brane universe}",
    eprint = "hep-th/0209076",
    archivePrefix = "arXiv",
    doi = "10.1016/S0370-2693(03)00555-0",
    journal = "Phys. Lett. B",
    volume = "562",
    pages = "118--126",
    year = "2003"
}

@article{Charmousis2002,
doi = {10.1088/0264-9381/19/18/304},
url = {https://doi.org/10.1088/0264-9381/19/18/304},
year = {2002},
month = {aug},
publisher = {},
volume = {19},
number = {18},
pages = {4671},
author = {Christos Charmousis and Jean-François Dufaux},
title = {General Gauss–Bonnet brane cosmology},
journal = {Classical and Quantum Gravity},
}

@article{Kuchar1976,
    author = "Kuchar, K.",
    title = "{Geometry of Hyperspace. 1.}",
    doi = "10.1063/1.522976",
    journal = "J. Math. Phys.",
    volume = "17",
    pages = "777--791",
    year = "1976"
}

@article{Teitelboim1973,
    author = "Teitelboim, Claudio",
    title = "{How commutators of constraints reflect the space-time structure}",
    doi = "10.1016/0003-4916(73)90096-1",
    journal = "Annals Phys.",
    volume = "79",
    pages = "542--557",
    year = "1973"
}

@article{ReggeTeitelboim1974,
  author = {Regge, Tullio and Teitelboim, Claudio},
  title = {Role of surface integrals in Hamiltonian formulation of general relativity},
  journal = {Annals of Physics},
  volume = {88},
  pages = {286--318},
  year = {1974},
  doi = {10.1016/0003-4916(74)90404-7}
}

@article{Carroll2001,
    author = "Carroll, Sean M.",
    title = "{The Cosmological constant}",
    eprint = "astro-ph/0004075",
    archivePrefix = "arXiv",
    reportNumber = "EFI-2000-13",
    doi = "10.12942/lrr-2001-1",
    journal = "Living Rev. Rel.",
    volume = "4",
    pages = "1",
    year = "2001"
}

@article{Weinberg1989,
  author = {Weinberg, Steven},
  title = {The cosmological constant problem},
  journal = {Reviews of Modern Physics},
  volume = {61},
  pages = {1},
  year = {1989},
  doi = {10.1103/RevModPhys.61.1}
}

@article{Isenberg1978,
  author  = {Isenberg, James A.},
  title   = {Constant mean curvature solutions of the Einstein constraint equations on closed manifolds},
  journal = {Class. Quantum Grav.}, 
  number  = {12},
  pages   = {2249–2274},
  year    = {1995},
  doi     = {https://doi.org/10.1088/0264-9381/12/9/013}
}

@article{HartleHawking1983,
  title = {Wave function of the Universe},
  author = {Hartle, J. B. and Hawking, S. W.},
  journal = {Phys. Rev. D},
  volume = {28},
  issue = {12},
  pages = {2960--2975},
  numpages = {0},
  year = {1983},
  month = {Dec},
  publisher = {American Physical Society},
  doi = {10.1103/PhysRevD.28.2960},
  url = {https://link.aps.org/doi/10.1103/PhysRevD.28.2960}
}

@book{Weinberg1972,
  author    = {Weinberg, Steven},
  title     = {Gravitation and Cosmology: Principles and Applications of the General Theory of Relativity},
  publisher = {John Wiley \& Sons},
  address   = {New York},
  year      = {1972}
}

@book{Maggiore2008,
  author    = {Maggiore, Michele},
  title     = {Gravitational Waves: Volume 1: Theory and Experiments},
  publisher = {Oxford University Press},
  address   = {Oxford},
  year      = {2008},
  doi       = {10.1093/acprof:oso/9780198570745.001.0001}
}

@book{Taylor2011,
  author    = {Taylor, Michael E.},
  title     = {Partial Differential Equations I: Basic Theory},
  edition   = {2},
  publisher = {Springer},
  address   = {New York},
  year      = {2011},
  doi       = {10.1007/978-1-4419-7055-8}
}

@book{MisnerThorneWheeler1973,
       author = {{Misner}, Charles W. and {Thorne}, Kip S. and {Wheeler}, John Archibald},
        title = "{Gravitation}",
         year = 1973,
       adsurl = {https://ui.adsabs.harvard.edu/abs/1973grav.book.....M}
}

@book{Carroll2004,
  author    = {Carroll, Sean M.},
  title     = {Spacetime and Geometry: An Introduction to General Relativity},
  publisher = {Addison-Wesley},
  address   = {San Francisco},
  year      = {2004}
}

@book{PoissonWill2014,
  author    = {Poisson, Eric and Will, Clifford M.},
  title     = {Gravity: Newtonian, Post-Newtonian, Relativistic},
  publisher = {Cambridge University Press},
  address   = {Cambridge},
  year      = {2014}
}

@book{Ringstrom2009,
  author    = {Ringström, Hans},
  title     = {The Cauchy Problem in General Relativity},
  publisher = {European Mathematical Society},
  address   = {Zürich},
  year      = {2009},
  doi       = {10.4171/053}
}

@article{Friedrich1985,
       author = {{Friedrich}, Helmut},
        title = "{On the hyperbolicity of Einstein's and other gauge field equations}",
      journal = {Communications in Mathematical Physics},
         year = 1985,
        month = dec,
       volume = {100},
       number = {4},
        pages = {525-543},
          doi = {10.1007/BF01217728},
       adsurl = {https://ui.adsabs.harvard.edu/abs/1985CMaPh.100..525F}
}

@inproceedings{York1979,
    author = "York, Jr., James W.",
    title = "{Kinematics and Dynamics of General Relativity}",
    booktitle = "{Workshop on Sources of Gravitational Radiation}",
    pages = "83--126",
    year = "1978"
}

@inproceedings{FischerMarsden1979,
  title={The initial value problem and the dynamical formulation of general relativity},
  author={Arthur E. Fischer and Jerrold E. Marsden},
  year={1979},
  publisher = {Cambridge University Press},
  pages = "138-211",
  url={https://api.semanticscholar.org/CorpusID:15669160}
}

@inproceedings{BartnikIsenberg2004,
    author = "Bartnik, Robert and Isenberg, Jim",
    title = "{The Constraint equations}",
    booktitle = "{50 Years of the Cauchy Problem in General Relativity: Summer School on Mathematical Relativity and Global Properties of Solutions of Einstein's Equations}",
    eprint = "gr-qc/0405092",
    archivePrefix = "arXiv",
    doi = "10.1007/978-3-0348-7953-8_1",
    year = "2002"
}

@book{Penrose1989,
  author    = {Penrose, Roger},
  title     = {The Emperor's New Mind: Concerning Computers, Minds, and the Laws of Physics},
  publisher = {Oxford University Press},
  address   = {Oxford},
  year      = {1989}
}

@article{CarrollChen2004,
    author = "Carroll, Sean M. and Chen, Jennifer",
    title = "{Spontaneous inflation and the origin of the arrow of time}",
    eprint = "hep-th/0410270",
    archivePrefix = "arXiv",
    reportNumber = "EFI-2004-33",
    month = "10",
    year = "2004"
}

@article{Joyce2015,
    author = "Joyce, Austin and Jain, Bhuvnesh and Khoury, Justin and Trodden, Mark",
    title = "{Beyond the Cosmological Standard Model}",
    eprint = "1407.0059",
    archivePrefix = "arXiv",
    primaryClass = "astro-ph.CO",
    doi = "10.1016/j.physrep.2014.12.002",
    journal = "Phys. Rept.",
    volume = "568",
    pages = "1--98",
    year = "2015"
}

@article{Copeland2006,
    author = "Copeland, Edmund J. and Sami, M. and Tsujikawa, Shinji",
    title = "{Dynamics of dark energy}",
    eprint = "hep-th/0603057",
    archivePrefix = "arXiv",
    doi = "10.1142/S021827180600942X",
    journal = "Int. J. Mod. Phys. D",
    volume = "15",
    pages = "1753--1936",
    year = "2006"
}

@article{RemmenCarroll2013,
  title = {Attractor solutions in scalar-field cosmology},
  author = {Remmen, Grant N. and Carroll, Sean M.},
  journal = {Phys. Rev. D},
  volume = {88},
  issue = {8},
  pages = {083518},
  numpages = {14},
  year = {2013},
  month = {Oct},
  publisher = {American Physical Society},
  doi = {10.1103/PhysRevD.88.083518},
  url = {https://link.aps.org/doi/10.1103/PhysRevD.88.083518}
}

@article{Freivogel2011,
    author = "Freivogel, Ben",
    title = "{Making predictions in the multiverse}",
    eprint = "1105.0244",
    archivePrefix = "arXiv",
    primaryClass = "hep-th",
    doi = "10.1088/0264-9381/28/20/204007",
    journal = "Class. Quant. Grav.",
    volume = "28",
    pages = "204007",
    year = "2011"
}

@article{Guth2007,
    author = "Guth, Alan H.",
    editor = "Sola, Joan",
    title = "{Eternal inflation and its implications}",
    eprint = "hep-th/0702178",
    archivePrefix = "arXiv",
    reportNumber = "MIT-CTP-3811",
    doi = "10.1088/1751-8113/40/25/S25",
    journal = "J. Phys. A",
    volume = "40",
    pages = "6811--6826",
    year = "2007"
}

@inbook{Ellis2012,
    author = "Ellis, George F. R.",
    editor = "Butterfield, Jeremy and Earman, John",
    title = "{Issues in the philosophy of cosmology}",
    booktitle = "{Philosophy of physics}",
    eprint = "astro-ph/0602280",
    archivePrefix = "arXiv",
    doi = "10.1016/B978-044451560-5/50014-2",
    pages = "1183--1285",
    year = "2006"
}

@article{Geroch1970,
    author = "Geroch, Robert P.",
    title = "{The domain of dependence}",
    doi = "10.1063/1.1665157",
    journal = "J. Math. Phys.",
    volume = "11",
    pages = "437--439",
    year = "1970"
}

@article{Senovilla1998,
    author = "Senovilla, Jos{\'e} M. M.",
    title = "{Singularity Theorems and Their Consequences}",
    eprint = "1801.04912",
    archivePrefix = "arXiv",
    primaryClass = "gr-qc",
    doi = "10.1023/A:1018801101244",
    journal = "Gen. Rel. Grav.",
    volume = "30",
    pages = "701",
    year = "1998"
}

@article{Lebowitz1993,
  author  = {Lebowitz, Joel L.},
  title   = {Boltzmann's Entropy and Time's Arrow},
  journal = {Physics Today},
  volume  = {46},
  pages   = {32--38},
  year    = {1993},
  doi     = {https://doi.org/10.1063/1.881363}
}

@book{Callender2010,
  author    = {Callender, Craig},
  title     = {The Oxford Handbook of Philosophy of Time},
  publisher = {Oxford University Press},
  address   = {Oxford},
  year      = {2010}
}

@article{GerochTraschen1987,
  title = {Strings and other distributional sources in general relativity},
  author = {Geroch, Robert and Traschen, Jennie},
  journal = {Phys. Rev. D},
  volume = {36},
  issue = {4},
  pages = {1017--1031},
  numpages = {0},
  year = {1987},
  month = {Aug},
  publisher = {American Physical Society},
  doi = {10.1103/PhysRevD.36.1017},
  url = {https://link.aps.org/doi/10.1103/PhysRevD.36.1017}
}

@article{SteinbauerVickers2006,
    author = "Steinbauer, Roland and Vickers, James A.",
    title = "{The Use of generalised functions and distributions in general relativity}",
    eprint = "gr-qc/0603078",
    archivePrefix = "arXiv",
    doi = "10.1088/0264-9381/23/10/R01",
    journal = "Class. Quant. Grav.",
    volume = "23",
    pages = "R91--R114",
    year = "2006"
}

@article{Arnowitt2008,
    author = "Arnowitt, Richard L. and Deser, Stanley and Misner, Charles W.",
    title = "{The Dynamics of general relativity}",
    eprint = "gr-qc/0405109",
    archivePrefix = "arXiv",
    doi = "10.1007/s10714-008-0661-1",
    journal = "Gen. Rel. Grav.",
    volume = "40",
    pages = "1997--2027",
    year = "2008"
}

@article{Gourgoulhon20072,
    author = "Gourgoulhon, Eric",
    title = "{3+1 formalism and bases of numerical relativity}",
    eprint = "gr-qc/0703035",
    archivePrefix = "arXiv",
    month = "3",
    year = "2007"
}

@article{Andersson2008,
    author = "Andersson, Lars and Moncrief, Vincent",
    title = "{Elliptic hyperbolic systems and the Einstein equations}",
    eprint = "gr-qc/0110111",
    archivePrefix = "arXiv",
    doi = "10.1007/s00023-003-0120-1",
    journal = "Annales Henri Poincare",
    volume = "4",
    pages = "1--34",
    year = "2003"
}

@book{Zeidler1986,
  author    = {Zeidler, Eberhard},
  title     = {Nonlinear Functional Analysis and Its Applications I: Fixed-Point Theorems},
  publisher = {Springer},
  address   = {New York},
  year      = {1985},
  isbn      = {9780387909141}
}

@article{IshamKuchar1985,
    author = "Isham, C. J. and Kuchar, K. V.",
    title = "{Representations of Space-time Diffeomorphisms. 1. Canonical Parametrized Field Theories}",
    reportNumber = "Imperial/TP/84-85/3",
    doi = "10.1016/0003-4916(85)90018-1",
    journal = "Annals Phys.",
    volume = "164",
    pages = "288",
    year = "1985"
}

@inbook{ChoquetBruhatYork1980,
    author = "Choquet-Bruhat, Y. and York, Jr., James W.",
    title = "{The Cauchy problem}",
    booktitle = "{Measurement of the Spin Dependent Total Cross Section $\Delta\sigma_L$ in $pp$ Collisions Between 200 and 600 MeV}",
    pages = "99--172",
    year = "1980"
}

@article{SmarrYork1978,
  title = {Kinematical conditions in the construction of spacetime},
  author = {Smarr, Larry and York, James W.},
  journal = {Phys. Rev. D},
  volume = {17},
  issue = {10},
  pages = {2529--2551},
  numpages = {0},
  year = {1978},
  month = {May},
  publisher = {American Physical Society},
  doi = {10.1103/PhysRevD.17.2529},
  url = {https://link.aps.org/doi/10.1103/PhysRevD.17.2529}
}

@article{BerezinKuzminTkachev1987,
  title = {Dynamics of bubbles in general relativity},
  author = {Berezin, V. A. and Kuzmin, V. A. and Tkachev, I. I.},
  journal = {Phys. Rev. D},
  volume = {36},
  issue = {10},
  pages = {2919--2944},
  numpages = {0},
  year = {1987},
  month = {Nov},
  publisher = {American Physical Society},
  doi = {10.1103/PhysRevD.36.2919},
  url = {https://link.aps.org/doi/10.1103/PhysRevD.36.2919}
}

@article{FrolovIsraelUnruh1990,
    author = "Frolov, Valeri P. and Israel, W. and Unruh, W. G.",
    title = "{Gravitational Fields of Straight and Circular Cosmic Strings: Relation Between Gravitational Mass, Angular Deficit, and Internal Structure}",
    doi = "10.1103/PhysRevD.39.1084",
    journal = "Phys. Rev. D",
    volume = "39",
    pages = "1084--1096",
    year = "1989"
}

@article{Hogg1999,
    author = "Hogg, David W.",
    title = "{Distance measures in cosmology}",
    eprint = "astro-ph/9905116",
    archivePrefix = "arXiv",
    month = "5",
    year = "1999"
}

@book{ChoquetBruhatDeWittMorette1989,
  author    = {Choquet-Bruhat, Yvonne and DeWitt-Morette, C{\'e}cile},
  title     = {Analysis, Manifolds and Physics},
  publisher = {North-Holland},
  address   = {Amsterdam},
  year      = {2000},
  doi      = {https://doi.org/10.1016/B978-0-444-50473-9.X5000-3}
}

@book{LadyzhenskayaUraltseva1968,
  author    = {Ladyzhenskaya, O. A. and Uraltseva, N. N.},
  title     = {Linear and Quasilinear Elliptic Equations},
  publisher = {Academic Press},
  year      = {1968},
  isbn      = {9780080955544}
}

@book{Aubin1982,
  author    = {Aubin, Thierry},
  title     = {Nonlinear Analysis on Manifolds. Monge--Amp\`ere Equations},
  publisher = {Springer-Verlag},
  address   = {New York},
  year      = {1982},
  volume    = {252},
  doi       = {https://doi.org/10.1007/978-1-4612-5734-9}
}

@book{McLean2000,
  author    = {McLean, William},
  title     = {Strongly Elliptic Systems and Boundary Integral Equations},
  publisher = {Cambridge University Press},
  address   = {Cambridge},
  year      = {2000},
  series    = {Cambridge University Press},
  doi       = {10.1017/CBO9780511606236},
  isbn      = {9780521663755}
}

@inproceedings{Halliwell1991,
    author = "Halliwell, Jonathan J.",
    title = "{Introductory Lectures on Quantum Cosmology}",
    eprint = "0909.2566",
    archivePrefix = "arXiv",
    primaryClass = "gr-qc",
    reportNumber = "MIT-CTP-1845",
    year = "1989"
}

@inproceedings{Kuchar1981,
    author = "Kuchar, K.",
    title = "{Canonical Methods of Quantization}",
    booktitle = "{Oxford Conference on Quantum Gravity}",
    year = "1980"
}

@article{Bartnik1988,
  author  = {Bartnik, Robert},
  title   = {Remarks on Cosmological Spacetimes and Constant Mean Curvature Surfaces},
  journal = {Commun. Math. Phys.},
  volume  = {117}, 
  pages   = {615--624},
  year    = {1988},
  doi     = { https://doi.org/10.1007/BF01218388}
}

@BOOK{EllisMaartensMacCallum2012,
       author = {{Ellis}, George F.~R. and {Maartens}, Roy and {MacCallum}, Malcolm A.~H.},
        title = "{Relativistic Cosmology}",
         year = 2012,
       adsurl = {https://ui.adsabs.harvard.edu/abs/2012reco.book.....E}
}

\end{document}